\documentclass[a4paper,USenglish,cleveref, autoref,
thm-restate]{lipics-v2021}

\usepackage{amsmath}
\usepackage{amsthm}
\usepackage{graphicx}
\usepackage[colorinlistoftodos]{todonotes}
\usepackage{algorithm}
\usepackage[noend]{algpseudocode}
\usepackage{enumerate}
\usepackage{xcolor}
\usepackage{framed}
\usepackage{bm}
\usepackage{pifont}
\usepackage{multicol}
\usepackage{lipsum}
\usepackage{tikz}
\usetikzlibrary{calc}
\usepackage{array}
\usepackage{relsize}
\usepackage{comment}
\usepackage{url}
\usepackage[title]{appendix}
\usepackage{caption}
\usepackage{subcaption}
\usepackage{xstring}
\usepackage{varwidth}
\usepackage{sansmath}
\usepackage{xspace}
\usepackage{mathtools}
\usepackage{thmtools}
\usepackage{thm-restate}
\usepackage[normalem]{ulem}

\usepackage{xifthen}
\usepackage{cleveref}

\newtheorem{ensemble}{Ensemble}
\newtheorem*{ensemble*}{Ensemble}
\newtheorem*{theorem*}{Theorem}

\algdef{SE}[RECEIVING]{Receiving}{EndReceiving}[1]{\textbf{upon
receiving}\ #1\ \algorithmicdo}{\algorithmicend\ \textbf{}}%
\algtext*{EndReceiving}

\algdef{SE}[UPON]{Upon}{EndUpon}[1]{\textbf{upon
}\ #1\ \algorithmicdo}{\algorithmicend\ \textbf{}}%
\algtext*{EndUpon}

\algdef{SE}[PROCEDURE]{Procedure}{EndProcedure}%
   [2]{\algorithmicprocedure\ \textproc{#1}\ifthenelse{\equal{#2}{}}{}{(#2)}}%
   {\algorithmicend\ \algorithmicprocedure}%

\newcommand{\lr}[1]{\langle #1 \rangle}
\newcommand{\multiset}[1]{\{\!\{#1\}\!\}}

\newcommand{\fz}{\hat{\newf}}
\newcommand{\newf}{\bm{f}}
\newcommand{\ignore}[1]{}
\newcommand{\ens}{\textsc{ens}}
\newcommand{\scn}{\ens}
\newcommand{\alg}{$\mathcal A$}

\newcommand{\threshold}{\emph {t}}
\newcommand{\tp}{\mathsf{T}}

\newcommand{\act}{\alpha}
\newcommand{\Child}{{\textsc Child}}
\newcommand{\ls}[1]{\emph{l$_{#1}$}}
\newcommand{\mult}[2]{\textit{mult}(#1,#2)}
\newcommand{\maxmult}[1]{\textit{max\_mult}(#1)}
\newcommand{\Vin}{\mathcal{V}_{\textit{in}}}
\newcommand{\Oracle}{{\tt p\_ob}}

\newcommand{\M}[2]{%
  \ifthenelse{\equal{#1}{#2}}%
    {[#1]}
    {[#1,#2]}
}

\newcommand{\optimal}[1]{%
  \ifthenelse{\equal{#1}{C}}%
    {Qualitative\xspace}
    {qualitative\xspace}
}

\def\HiLi{\leavevmode\rlap{\hbox to
\hsize{\color{gray!10}\leaders\hrule height .8\baselineskip
depth .5ex\hfill}}}

\algblockdefx[ON]{On}{EndOn}[1]
  {{\bf On} #1\   \algorithmicdo }
  {\algorithmicend\ }
\makeatletter
\ifthenelse{\equal{\ALG@noend}{t}}%
  {\algtext*{EndOn}}
  {}%
\makeatother

\makeatletter

\tikzset{%
  remember picture with id/.style={%
    remember picture,
    overlay,
    save picture id=#1,
  },
  save picture id/.code={%
    \edef\pgf@temp{#1}%
    \immediate\write\pgfutil@auxout{%
      \noexpand\savepointas{\pgf@temp}{\pgfpictureid}}%
  },
  if picture id/.code args={#1#2#3}{%
    \@ifundefined{save@pt@#1}{%
      \pgfkeysalso{#3}%
    }{
      \pgfkeysalso{#2}%
    }
  }
}

\def\savepointas#1#2{%
  \expandafter\gdef\csname save@pt@#1\endcsname{#2}%
}

\def\tmk@labeldef#1,#2\@nil{%
  \def\tmk@label{#1}%
  \def\tmk@def{#2}%
}

\tikzdeclarecoordinatesystem{pic}{%
  \pgfutil@in@,{#1}%
  \ifpgfutil@in@%
    \tmk@labeldef#1\@nil
  \else
    \tmk@labeldef#1,(0pt,0pt)\@nil
  \fi
  \@ifundefined{save@pt@\tmk@label}{%
    \tikz@scan@one@point\pgfutil@firstofone\tmk@def
  }{%
  \pgfsys@getposition{\csname save@pt@\tmk@label\endcsname}\save@orig@pic%
  \pgfsys@getposition{\pgfpictureid}\save@this@pic%
  \pgf@process{\pgfpointorigin\save@this@pic}%
  \pgf@xa=\pgf@x
  \pgf@ya=\pgf@y
  \pgf@process{\pgfpointorigin\save@orig@pic}%
  \advance\pgf@x by -\pgf@xa
  \advance\pgf@y by -\pgf@ya
  }%
}

\newlength\AlgIndent
\newcounter{mymark}
\newcommand\ColorLine{%
  \stepcounter{mymark}%
  \tikz[remember picture with id=mark-\themymark,overlay] {;}%
  \begin{tikzpicture}[remember picture,overlay]%
    \filldraw[gray!20]%
   let \p1=(pic cs:mark-\themymark),
   \p2=(current page.east)  in
   ([xshift=-\ALG@thistlm-0em,yshift=-0.7ex]0,\y1)  rectangle
   ++(\linewidth+\AlgIndent,\baselineskip); \end{tikzpicture}%
}%

\newcommand\ColorLinex{%
  \stepcounter{mymark}%
  \tikz[remember picture with id=mark-\themymark,overlay] {;}%
  \begin{tikzpicture}[remember picture,overlay]%
    \filldraw[gray!20]%
   let \p1=(pic cs:mark-\themymark),
   \p2=(current page.east)  in
   ([xshift=-\ALG@thistlm--3em,yshift=-0.7ex]0,\y1)
   rectangle ++(\linewidth+\AlgIndent,\baselineskip); \end{tikzpicture}%
}%

\algnewcommand\CREQUIRE{\item[\setlength\AlgIndent{1.6em}\ColorLine\algorithmicrequire]}%
\algnewcommand\CENSURE{\item[\setlength\AlgIndent{1.6em}\ColorLine\algorithmicensure]}%
\algnewcommand\CSTATE{\State\ColorLine}%
\algnewcommand\CSTATEx{\Statex\ColorLinex}%
\algnewcommand\CCOMMENT{\Comment\ColorLine}%

   \algdef{SE}[IF]{CIF}{ENDIF}%
   [2][default]{\ColorLine\algorithmicif\ #2\ \algorithmicthen}%
   {\algorithmicend\ \algorithmicif}%

   \algdef{SE}[FOR]{CFOR}{ENDFOR}%
   [2][default]{\ColorLine\algorithmicfor\ #2\ \algorithmicdo}%
   {\algorithmicend\ \algorithmicfor}%
\makeatother

\nolinenumbers

\title{Probabilistic Indistinguishability and the\\ Quality of Validity in Byzantine Agreement}

\titlerunning{Probabilistic Indistinguishability \&  Qualitative Validity} 

\author{Guy Goren}{Technion}{}{}{}
\author{Yoram Moses}{Technion}{}{}{}
\author{Alexander Spiegelman}{Novi Research}{}{}{}

\authorrunning{G. Goren, Y. Moses, and A. Spiegelman} 


\begin{document}

\maketitle

\begin{abstract}
	Lower bounds and impossibility results in distributed computing are both intellectually challenging and practically important.
	Hundreds if not thousands of proofs appear in the literature, but surprisingly, the vast majority of them apply to deterministic algorithms only.
	Probabilistic protocols have been around for at least four decades and are receiving a lot of attention with the emergence of blockchain systems.
	Nonetheless, we are aware of only a  handful of randomized lower bounds.
	
	In this paper we provide a formal framework for reasoning about randomized distributed algorithms.
	We generalize the notion of indistinguishability, the most useful tool in deterministic lower bounds, to apply to a probabilistic setting.
	We apply this framework to prove a result of independent interest. Namely, we completely characterize  the quality of decisions that protocols for a randomized  multi-valued Consensus problem can guarantee in an asynchronous environment with Byzantine faults.
	We use the new notion to prove a lower bound on the probability at which it can be guaranteed that honest parties will not decide on a possibly bogus value. Finally, we show that the bound is tight by providing a protocol that matches it. 
\end{abstract}

\section{Introduction}
\label{sec:intro}
Randomized  algorithms have a long tradition in distributed computing~\cite{rabin1980probabilistic}, 
where they have been applied to many different problems in a variety of models~\cite{lynch1996distributed}.
In the context of fault-tolerant agreement they have served to overcome the  impossibility of 
agreement in asynchronous settings
\cite{FLP85,rabin1983randomized,ben1983another}, and have
significantly improved efficiency compared to deterministic
solutions~\cite{feldman1997optimal, katz2006expected}.
With the recent prevalence  of blockchain systems, Byzantine agreement algorithms that can overcome malicious parties have found renewed interest in both industry and academia.
For obvious reasons, blockchain systems should strive to minimize the share of decisions that originate from malicious parties, and to increase the share originating from honest ones. A natural question, then, is 
what are the inherent limits on the quality of Byzantine
agreement algorithms in this regard? Namely, what can we say about the
probability with which an algorithm can guarantee that a good decision is made?

Given their practical importance,  characterizing the power
and limitations of randomized  distributed algorithms for agreement has
become ever more desirable.
However, obtaining tight, or nontrivial, probabilistic bounds on
properties in the asynchronous Byzantine setting can be a challenging
task.
As is well known, there are
``{\it Hundreds of impossibility results for distributed
computing}''~\cite{fich2003hundreds}. But very few of them apply to randomized
protocols.  
Unfortunately, there is currently a dearth of general tools for
characterizing the properties of randomized algorithms.

The notion of indistinguishability has for years been one of
the most useful tools for proving deterministic lower bounds and
impossibility results in distributed
computing~\cite{attiya2014impossibility}.
Such deterministic lower bounds typically rely on the fact that
if a correct party cannot distinguish between two executions of a
deterministic protocol (i.e., its local state is the same in both),
then it performs the same actions in both.
In a randomized algorithm, the fact that two executions are
indistinguishable to a given party up to a certain time does not ensure
that its next action will be the same in both. Moreover, a single
execution does not, by itself, provide information on the probability with which actions are performed.
As a result, the classic notion of indistinguishability does not directly capture many of the probabilistic aspects of a randomized algorithm. 

Of course, probabilistic properties of distributed algorithms such as 
	``{\it the probability that the parties decide on a value proposed by an honest party is at least~$x$}''
or ``{\it all honest parties terminate with
		probability 1}'' cannot be evaluated based on an individual execution.
Clearly, to make formal sense of such statements, we need to define an appropriate  probability space. However, due to the nondeterminism inherent in our model, a probability space over the set of all executions cannot be defined (cf.~\cite{Pnueli83}). This is because  we can't assume a distribution over the initial configurations, and similarly there is no well-defined distribution on the actions of the adversary, who is in charge of all the nondeterministic decisions.   
Once we fix the adversary's strategy, we are left with a purely probabilistic structure, which we call an {\em ensemble}. 
An ensemble naturally induces a probability space.
This allows us to formally state probabilistic properties of an algorithm \alg{} of interest with respect to all of its ensembles ($=$~adversary strategies). E.g., {\it ``for every ensemble of algorithm~\alg, all honest parties terminate with	probability~1.''}


In deterministic algorithms,  indistinguishability among executions is determined based on a party $p_i$'s local history, i.e., the sequence of local states that~$p_i$ passes through in the executions. 
We generalize the notion of  an $i$-local history to a notion called an {\it $i$-local ensembles}. A local ensemble is a tree of local states, that captures subtle, albeit essential, aspects of probabilistic protocols. 
This facilitates the definition of a notion of \textbf{\textit{probabilistic indistinguishability}} among ensembles, whereby 
two ensembles are considered indistinguishable to a process~$p_i$ if they induce identical $i$-local ensembles. 
Indistinguishability among ensembles provides a formal and convenient framework that 
can be used to simplify existing lower bound proofs in a probabilistic setting, and to prove new ones. 
A significant feature of this framework is its simplicity and ease of use, allowing similar arguments as in the deterministic case. 
The notions contain just enough structure beyond that of their deterministic analogues to capture the desired probabilistic properties. 

Our original motivation for developing the above framework was to formally prove 
tight probabilistic bounds on  the share of good decisions made by a  randomized 
Byzantine agreement algorithm in an asynchronous  setting. 
In
\Cref{sec:Validity} we use probabilistic indistinguishability to prove that, roughly
speaking, no algorithm can guarantee that the probability  to decide on
a genuine input value  is greater than $1 - \frac{f}{n-t}$. (As usual, $n$ is
the total number of parties here, while~$t$ and~$f$ are the maximal and actual
number of failures, respectively.) 
Moreover,  this bound is shown to be tight, by presenting an algorithm that achieves it.

This paper makes two distinct and complementary main contributions:
\begin{itemize}
\item We define a notion of  indistinguishability that generalizes its deterministic counterpart, and is suitable for proving lower bounds in the context of probabilistic protocols.
	A new element in our definition is a purely probabilistic tree whose paths represent local histories of a given process.
	The resulting framework provides an	intuitive and rigorous way to reason about probabilistic properties of such protocols.
 
 \item We introduce \emph{\optimal{C} Validity}, a new probabilistic  validity condition for the Byzantine agreement problem. It provides a probabilistic bound on the ability of corrupt parties to bias the decision values, which is of interest in the blockchain arena. We prove that, in a precise sense, it is the strongest achievable validity property in the  asynchronous setting. Both the statement of the property and the proof are faciltated by our  new framework. 
 
\end{itemize}

%

\section{Related Work}
\label{sec:related}

\noindent\textbf{Deterministic indistinguishability in distributed computing.}
The main inherent limitation of distributed computing
leveraged by most proofs is the lack of global
knowledge~\cite{lynch1989hundred, fagin2003reasoning}. 
That is, each party needs to evaluate the global state of the system
based on its local state and act accordingly.
Deterministic indistinguishability captures exactly that: if
two executions of a deterministic protocol looks the same from some
party's local point of view, then this party performs exactly the same
sequence of actions in both.
Good surveys of techniques used in these proofs are
presented in~\cite{attiya2014impossibility, fich2003hundreds, lynch1989hundred}. 
They all utilize indistinguishable executions, but differ in the way
they construct them.\\ 

\noindent\textbf{Lower bound approaches for randomized distributed protocols.}
There are several lower bound and impossibility results for randomized
algorithms in the literature.
One approach is to reduce the distributed problem into a cleaner
mathematical one that abstracts away the issues of local knowledge, and
then apply methodologies from other fields to the latter
problem~\cite{aspnes1998lower, kushilevitz1993lower, bar1998tight,
attiya2008tight}.
An example of such a mathematical tool, used in~\cite{aspnes1998lower,
bar1998tight, attiya2008tight}, is a form of collective coin-fipping
game~\cite{ben1985collective}, which is an algorithm for combining
local coins into a single global one, whose bias should be small.
Each of the participants flips local coins and submits them to an
adversary. The adversary then gets to decide which coins to reveal and which to hide. 
The adversary's purpose is to bias the global coin's output, while
hiding as few local flips as possible.
Lower bounds on the number of coin flips required to tolerate an
adversary as a function of his budget to hide coins are proven
in~\cite{aspnes1998lower, bar1998tight}.
They relate these bounds to 
randomized distributed algorithms in the following way:
Intuitively, each step of the coin-flipping game corresponds to an
execution of the distributed algorithm up to some random event, which
can be interpreted as the flipping of a local coin.
The adversary's choice to hide or reveal this local coin corresponds to
its power to kill the process that executes the random
event or to let it run.

Aspnes~\cite{aspnes1998lower} extended the valency arguments introduced
in FLP~\cite{FLP85} to the probabilistic setting and used coin flipping
games to prove an $\Omega(\frac{n^2}{\log^2(n)})$ lower bound on
the expected step complexity of solving asynchronous randomized
Consensus.
Bar-Joseph and Ben-Or~\cite{bar1998tight} applied this technique to the
synchronous setting and proved an $\Omega(\sqrt{n/\log(n)})$ lower bound
on the number of rounds required for Consensus, in expectation, under a
worst case adversary.
Attiya and Censor-Hillel~\cite{attiya2008tight} closed a gap left in
Aspnes's work~\cite{aspnes1998lower} by showing that $\Theta(n^2)$ is a
tight bound on the step complexity for Consensus in the shared memory model.
For the lower bound, they elegantly combined valency arguments and coin
flipping games with a layering technique introduced
in~\cite{moses1998unified} for the deterministic case.
They restricted their adversarial strategies to proceed in layers of at
least $n-t$ parties. 
That is, for any given configuration, parties flip local coins and then 
the adversary picks at least $n-t$ parties and lets each of them 
perform a step, which collectively leads to the next configuration.
To capture the algorithm's randomness, they compute the probability to
decide $v$ from some configuration under a given adversary by
summarizing products of probability spaces induced by reachable
configurations in executions in which $v$ is decided.
To define the valence of a configuration under a set $S$ of possible
adversarial strategies, they check whether, starting from the given configuration, 
there are adversaries in
$S$ that lead to decisions with a probability higher than some threshold.
While the elegant way in which they treat the randomness of Consensus
algorithms allows them to distinguish among \emph{bivalent},
\emph{v-valent}, and \emph{null-valent} configurations as required by
their proof, it is not clear how their formalization can be
conveniently applied to other problems.

Another approach to proving randomized distributed lower bounds is via reducing 
the argument to deterministic indistinguishability by considering a
fixed random tape of coin flips to abstract away the randomness from an
execution~\cite{attiya2010lower, chor1989simple}.
An execution is deterministically defined by the initial configuration,
scheduler, and a fixed tape of coin flips.  
Given two executions, the standard indistinguishability argument can be
applied.
Attiya and Censor-Hillel~\cite{attiya2010lower}
proved a trade-off between termination probability and step
complexity of randomized Consensus algorithms in asynchronous systems.
To prove their result, they showed that for any random tape there is an
indistinguishability chain that starts and ends with executions that do
not allow the same decision value due to validity.
Assuming a minimal probability for termination per scheduler and
initial configuration, according to the tapes distribution, they show
that for at least one tape all the executions in the chain terminate.
By the indistinguishability chain (and the agreement condition) the decision value on
the two ends of the chain is the same, which contradicts validity.
Their argument can be restructured  within our framework
in a manner that perhaps brings their proof closer to the intuition.
An ensemble gathers all possible tapes for a given
adversarial strategy (i.e., initial configuration and scheduler).
As a result, instead of constructing an indistinguishability chain for
every random tape, we can construct a single indistinguishability chain
among ensembles. Moreover, instead of explicitly refering  to the termination
probability induced by the tapes' distributions, we can directly use the
termination probability defined on ensembles.

A recent work~\cite{abraham2019communication} extends a lower bound by
Dolev and Reischuk~\cite{dolev1985bounds} on the communication
complexity of Byzantine agreement to the randomized case.
Their proof makes use of arguments about the indistinguishability between two adversaries, without providing a formal definition.  In \Cref{sec:formalize} we show that using our definitions can fill this gap in their presentation. 
The work of Fich, Herlihy and Shavit in \cite{fich1998space}
proves space lower bounds for randomized shared objects.
Their technique exploits covering arguments, which, in turn, rely
on deterministic indistinguishability. 
To this end, they remove the stochastic nature of the problem by
considering the ``\emph{nondeterministic solo termination}''
property that requires only a non-zero probability for solo
termination. 
They also show that their bounds immediately apply
to nondeterministic wait-free objects, but conjecture that the bounds
are not tight for that case. 
A possible reason could be that their
technique ignores the quantitative probabilistic nature of the problem.
A recent paper by Ellen, Gelashvili, and
Zhu~\cite{ellen2018revisionist} shows that nondeterministic solo
termination and obstruction freedom are equivalent in the specific
context of space lower bounds.\\

\noindent\textbf{Indistinguishability and equivalence in probabilistic systems.}
Probabilistic protocols have a long tradition in the computer
science.
There is a broad literature concerned with rigorous
probabilistic analysis of such protocols, both in cryptography
\cite{goldreich2009foundations, Maurer02,goldwasser1989knowledge,yao1986generate,GMW87} and in distributed
systems \cite{SegalaPHD95, aspnes1998lower,Pnueli83,HartSharirPnueli83}.
Moreover, notions of indistinguishability play an important role in the
analysis of probabilistic systems. In cryptography, for example,
notions of computational indistinguishability and statistical
indistinguishability are routinely used in order to capture the fact
that protocols do not unintentionally leak information (e.g., in
zero-knowledge protocols and multi-party computation) and to formalize
notions such as pseudo-randomness. Indistinguishability in this context
is defined in terms of the difference between families of
distributions, and in terms of an agent's ability to tell them apart.
In the context of probabilistic automata there are well established
notions of simulation and bisimulation that provide definitions of
equivalence between two systems \cite{SegalaLynch95, CanettiCKLLynchPSegala06}.
These facilitated  notions of refinement and the verification of probabilistic
systems.

It may be possible to formulate probabilistic arguments used in lower
bound proofs for standard distributed systems protocols either in terms
of the cryptographic notions of indistinguishability or in terms of
bisimulation among I/O automata.
However, this would require nontrivial technical adjustments, and it is
not clear that it would provide new insights or
inroads into the essence of the proofs at hand.
Indeed, as reviewed above, none of the lower bound proofs on probabilistic protocols in distributed computing
that we are aware of make use of these frameworks.

\section{Model}
\label{sec:model}

We consider a standard message passing model with a set $\Pi$ of $n$
parties and an adversary~\cite{attiya2004distributed}.
Parties communicate via an asynchronous network of peer-to-peer communication links.
\footnote{Our definitions can be easily translated to synchronous
	communication and shared memory models.} 
Each party maintains a well-defined \emph{local state} at all
times.
We assume for simplicity that the local state of each party $p_i$ at 
a particular  time $\tp$ consists of an initial state $\ls{i}\!^0$ and the finite
sequence of local events at~$p_i$ up to time~$\tp$.
This sequence is composed of the actions that $p_i$ has performed
(including the messages it has sent) before time $\tp$, as well as the
messages that  $p_i$ received until time  $\tp$.
In particular, its \emph{initial} local state is $\lr{\ls{i}\!^0, []}$.
(For example, in a Consensus algorithm the initial local state of each party
is $\lr{v_i,[]}$, where~$v_i$ is its input value.)
A \emph{configuration} is a mapping from parties to their local states and from communication links to the set of pending messages therein. (A message that has been sent on a link but not yet delivered is \emph{pending}.)
An \emph{initial configuration} associates with each party its initial local state and
each link with an empty set of pending messages. 
An \emph{algorithm} defines the actions that each party performs (local computations, decisions and message sends) as a function of its local state.

We assume an interleaving model where at each point in time a single local event occurs \cite{lynch1996distributed}. 
A local event consists either of a local step performed by a party according to the algorithm, or of the delivery of a pending message.
The identity of the party that moves, or the pending message that is delivered, are determined by the adversary.
Both the scheduling of local steps and the delivery of messages are asynchronous.
I.e., while the adversary must schedule every correct party to move infinitely often, the relative rates by which parties move can be arbitrary.
Moreover, while every message sent must be eventually delivered (exactly once), there is no bound on how long messages spend in transit.
An \emph{execution} of a deterministic algorithm \alg{} is a (finite or infinite) sequence of the form $e=\langle C_0,\phi_1,C_1,\phi_2,C_2,\phi_3,\ldots\rangle $, where $C_0$ is an
initial configuration, $C_k$ is a configuration, and $\phi_k$ is either a local step or a message delivery for every $k>0$.
In case $\phi_k$ is a local step by a correct party, the configuration $C_{k}$ is obtained from $C_{k-1}$ by modifying this party's local state (and possibly an outgoing link) according to the algorithm~\alg.
If $\phi_k$ is the delivery of a message $m$ to party~$p_i$ from party $p_j$, then $C_{k}$ is similar to $C_{k-1}$ except that $m$ is removed from the link between $p_i$ and $p_j$ and appended to~$p_i$'s local state.

We also wish to model settings in which failures can occur.
In this case, the identity of faulty parties and their behavior are determined by the adversary.  
In any given execution, we associate with the adversary a \emph{strategy}, which determines all of its decisions in the execution. 
Various failure assumptions exist in the literature (e.g., crash, Byzantine, authenticated Byzantine, etc.).
Each failure model induces its own set of constraints on how the adversary's strategy affects failures.

In order to facilitate the study and analysis of randomized algorithms, we slightly extend the model by adding probabilistic objects (\Oracle{} for short) to the system.
We add a new type of local action (local step) that consists of a party accessing a probabilistic object. The result of this action is that the party immediately receives a return value from the object. 
A probabilistic object has a local state that can change following an access.
The return value obtained from accessing such an object is sampled from a given distribution, which may depend on the object's local state. The range of return values may be infinite but it must be countable.
A randomized algorithm is associated with the set of \Oracle s that it employs, and the specifications of  these objects are  part of the algorithm's definition.
Moreover, a configuration now contains the local states of the probabilistic objects, in addition to the state of the parties and the communication links.

We can, for example, model randomization via a simple local coin by
having a (stateless) probabilistic object that returns 1 or 0 with
probability 1/2 for each access.
Probabilistic objects can be used to model more complex situations in which
there may be correlations among values received by different parties.
The reason we use probabilistic objects is to facilitate the analysis
of systems with e.g., shared coins, VRFs, etc.\ all within the same
framework~\cite{king2016byzantine, aspnes1998lower, canetti1993fast,
micali1999verifiable, cachin2005random}.

\section{Probabilistic Indistinguishability}
\label{sec:probInd}

In this section, we generalize the notion of indistinguishability to account for probabilistic aspects.
First, let us review a standard definition of indistinguishability.
Recall that an algorithm determines a party's behavior as a function of its local state.
Therefore, two executions $e_1$ 
and $e_2$ 
are indistinguishable to a party $p_i$ if the sequence of local states it goes through in both is the same.
More formally, we define the \emph{$i$-local history} for a party $p_i$ in a given execution $e$ to be the (stuttering-free) sequence of $p_i$'s local states in  the configurations of $e=\langle C_0,\phi_1,C_1,\phi_2,C_2,\phi_3,\ldots\rangle$. I.e., the sequence of local states except that consecutive repetitions are removed.
Executions $e_1$ and $e_2$ are considered  \emph{indistinguishable} to $p_i$ if both executions induce the same $i$-local history.

In the context of randomized algorithms, one is often interested in probabilistic properties of the algorithm, such as
the probability that a given action is taken.
Since a single execution does not, in itself, contain such probabilistic information, 
indistinguishability between pairs of executions is not the appropriate notion for reasoning about such properties.
 Such reasoning requires assigning probabilities to events consisting of appropriately chosen sets of executions. Indeed, a probabilistic notion of  indistinguishablity should be based on relating such sets of executions.

There is typically no way to define a probability space over the set of all executions of a randomized algorithm in an asynchronous system with failures. The initial configuration, the scheduling of parties to move and of message deliveries, and the identities and behavior of faulty parties are considered genuinely nondeterministic decisions. 
No probabilistic distribution is assumed on these decisions. We
typically consider the nondeterministic decisions to be governed by an {\it adversary}. Adversaries come in different types, e.g., static or adaptive, and their abilities may vary, e.g., in terms of the type of failures that they can cause---crash vs.\ Byzantine, etc. The notion of indistinguishability that will be presented shortly applies to all of them; it  is independent of the type of adversary under consideration.
An adversary of a given type can employ many different concrete strategies. 
In all cases, fixing the adversary's strategy eliminates all nondeterminism. Consequently, every  transition is then either deterministic or purely probabilistic, and this induces a well-defined probability space over subsets of the executions. 
In order to define the probability space of interest w.r.t.\ a given strategy of the adversary, we proceed as follows.

\subsection{Ensembles}\label{sec:ensemble}

We define an \emph{ensemble} to be a directed weighted tree in which each node is a configuration, and edges represent local events.
When clear from the context, we may slightly abuse notation by writing $v\in\scn$ to denote that node~$v$ appears in the ensemble \ens.
For a node $v\in\scn$ we define $E_{\emph{out}}(v)$ to be the set of edges connecting $v$ to its children in \ens. 
We define \textit{maximal path} to be a path that cannot be extended, i.e., starts at the root and ends at a leaf. 
For readability, when clear from context, we refer to them simply as paths.
An ensemble \ens{} with respect to an algorithm \alg{} must satisfy the following properties:

\begin{itemize}  
  
  \item The root of the ensemble (tree) is an initial configuration $C_0$.

 \item Each path in the ensemble consists of the sequence of configurations of a legal execution of \alg. 

 \item For every node $v\in\ens$, the weights on the edges of $E_{\emph{out}}(v)$ are positive and their sum is~1. Moreover,
  
  \begin{itemize}
    
    \item If $|E_{\emph{out}}(v)| = 1$, then the edge $(v,u) \in
    E_{\emph{out}}(v)$ represents a deterministic local event.
         
    \item Otherwise, (when $|E_{\emph{out}}(v)| > 1$), the edges in $E_{\emph{out}}(v)$ represent a single local event consisting of accessing a probabilistic object. 
    Each edge $(v,u) \in E_{\emph{out}}(v)$ represents a possible return value, where the weight on an edge is the probability that the object access returns this value. 
  \end{itemize}  

\item The return values and the weights of edges that represent accesses to a probabilistic object~\Oracle{} in \ens{} satisfy the object's specifications.
For example, consider an object \Oracle{} consisting of a biased coin returning $1$ with probability $x\in[0,1]$ and returning $0$ with probability $1-x$. If a party accesses \Oracle{} at a configuration $v\in\ens$, then the node~$v$ has two children. 
One child, at the end of an edge labeled $x$ corresponds to the coin returning~$1$. The other child corresponds to the coin returning~$0$, at the end of an edge labeled~$1-x$.

\end{itemize} 

We note that similarly to how, in deterministic algorithms, an adversary's strategy determines an execution of \alg, in randomized algorithms the adversary's strategy determines an ensemble of \alg.
Ensembles generalize executions in the sense that a
deterministic algorithm yields ensembles that consist of a single path.

\paragraph*{Associating probabilities with configurations}
An ensemble \ens{} induces a probability space defined by the triplet $(\Omega_{\scn}, {\mathcal F}, P_{\scn})$ which is specified as follows.
$\Omega_{\scn}$ is the set of paths (executions) in \scn.
For each node~$v$ in the ensemble, define by $S_v$ the set of executions
(paths) in $\scn$ that pass through~$v$.
${\mathcal F}$ is the sigma algebra generated by $\{ S_v: v\in \scn \} $. (Closing under complement and countable unions.)
Finally, $P_{\scn}$ is the probability function defined by $P_{\scn}(S_v)\triangleq$ the product of the edge weights along the path from the root to~$v$.
As required, our definitions satisfy both $P_{\scn}(\bar{S})=1-P_{\scn}(S)$ and that $P_{\scn}$ is countably additive.

\subsection{Local Ensembles}\label{sec:local exp}
Our goal will be to define a notion of indistinguishability among ensembles, with respect to a particular party~$p_i$. 
Roughly speaking, this is determined by~$p_i$'s local histories in these ensembles. To this end, we consider a probabilistic tree consisting of the $i$-local histories in a given ensemble~$\scn$. This is called an \emph{$i$-local ensemble},  and is denoted by~$\scn_i$.  Since~$p_i$ can have the same local history in different paths of~$\ens$, a node in $\ens_i$ typically corresponds to several nodes of~$\ens$. Consequently, the construction of~$\ens_i$ must be done carefully, to correctly account for the probabilistic transitions in~$\ens_i$. 

For each unique local state $\ls{i}$  of~$p_i$ that appears in
$\scn$, there is a node $v_\ls{i}\in\scn_i$ labeled with the tuple
$\lr{\ls{i},p_\ls{i}}$, where $p_\ls{i}\in [0,1]$ is the probability of
$p_i$ to reach $\ls{i}$ in ensemble \scn.
More formally, for a local state $\ls{i}$, define 
$S_{\ls{i}}=\{S_v\mid
v\in\scn{} \text{ and $v$ is a configuration containing $\ls{i}$} \}$.
The resulting probability assigned to $v_{\ls{i}}\in\scn_i$ labeled with $\ls{i}$ is $p_\ls{i} = P_{\scn}(\bigcup\limits_{S_v\in S_{\ls{i}}} S_v)$.
The root of $\scn_i$ is labeled with $\ls{i}\!^0$, i.e., $p_i$'s local state at the root of~$\scn$, and~$p_{\ls{i}\!^0}=1$.
For each node $v_\ls{i}\in\scn_i$ labeled with an $i$-local state
$\ls{i}$, we define the set $\Child(v_\ls{i})$ to consist of the nodes $u_{\hat{\ls{i}}} \in\scn_i$ labeled by local states
$\hat{\ls{i}}\ne\ls{i}$ that directly follow
$\ls{i}$ in an execution contained in \scn.
I.e., there are $v,u\in\scn$ such that $u$ is a child of~$v$, and
$v,u$ represent configurations that contain $\ls{i}$ and
$\hat{\ls{i}}$, respectively.
(Notice that since local states contain
the full history of local events, once $p_i$ has transitioned from a
local state $\ls{i}$, this state never repeats.)
An illustrative example of an ensemble and its induced local ensemble for some party $p_i$ is given in \Cref{fig:ens and i-local exp}.
We remark that constructing an $i$-local ensemble $\scn_i$ is linear in the size of
\scn. An algorithm is provided in \Cref{sec:pseu...} for completeness.

\begin{figure*}[th]
	\centering
	\begin{subfigure}[t]{0.48\textwidth}
		\centering
		\includegraphics[width=\textwidth]{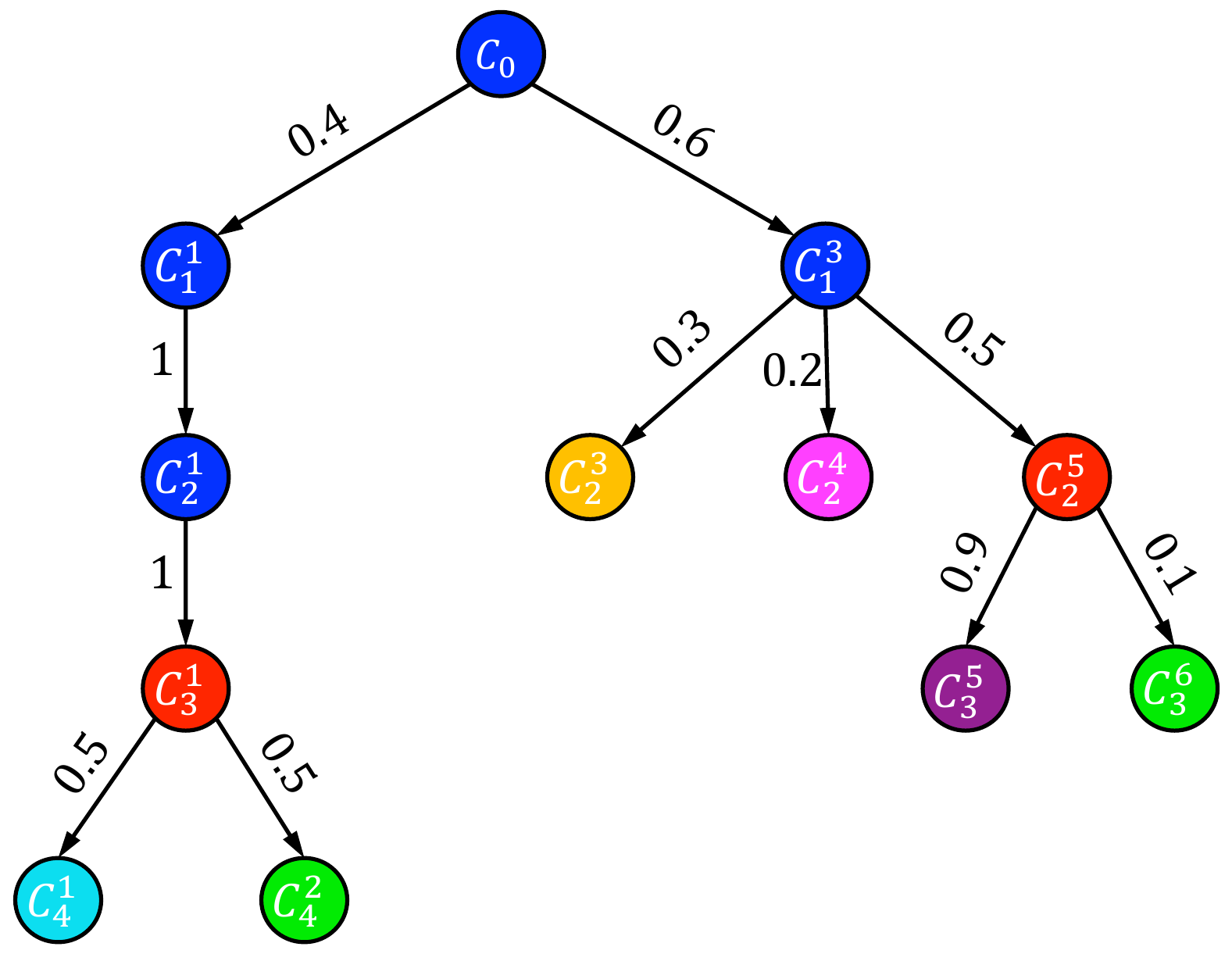}
		\caption{An ensemble \ens.}
	\end{subfigure}\hfill
	\begin{subfigure}[t]{0.48\textwidth}
		\centering
		\includegraphics[width=\textwidth]{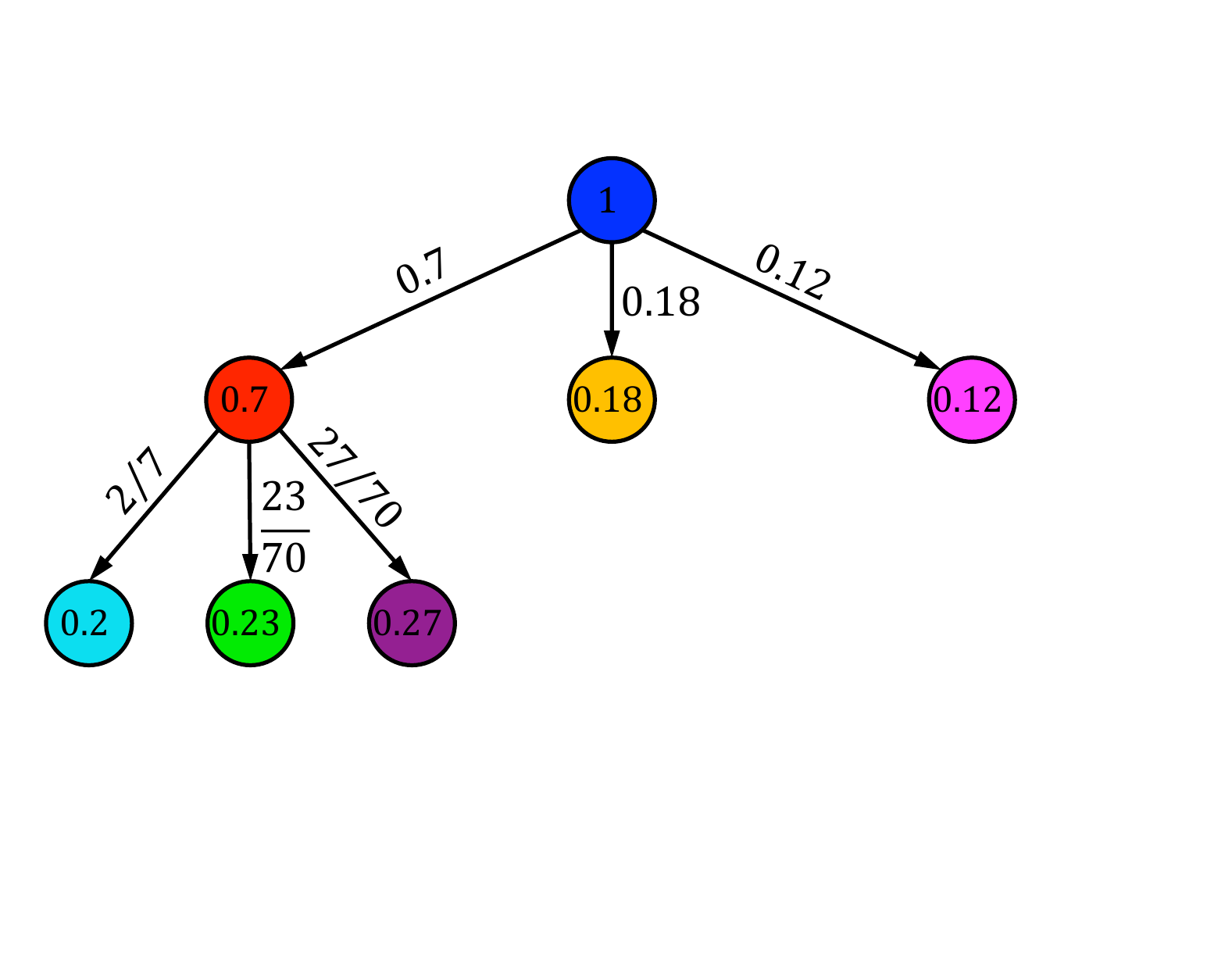}
		\caption{The $i$-local ensemble \ens$_i$ in~$\ens$.}
	\end{subfigure}
	\caption{An ensemble \ens{} and its induced $i$-local ensemble for some party~$p_i$. 
		The ensemble corresponds to a given adversary strategy. It 
		contains $6$ possible executions, all start in in configuration~$C_0$.
		Each node of \ens{} represents a unique configuration and the local state of party $p_i$ in this configuration is color encoded in our figure. For example, 
		$p_i$ has the same local state, which encode in red,  in $C_3^1$ and $C_2^5$.
		 The edges of \ens{} are labeled with their probabilities.
		 %
		 The color of a node $v_\ls{i}$ in the $i$-local ensemble $\ens_i$ represents its local state~$\ls{i}$ and the node is labeled with the probability $p_\ls{i}$ of reaching this local state under the current strategy. For convenience, edges of $\ens_i$ are weighted in a similar fashion to the edges of \ens.
	 }
	\label{fig:ens and i-local exp}
\end{figure*}

 \begin{definition}[\textbf{\textit{Probabilistic Indistinguishability}}]
 \label{def:PI}
 	Two ensembles are indistinguishable to a party~$p_i$ if they induce the same $i$-local ensemble.
 \end{definition}

For a given ensemble \scn{} and a local action~$\act$ of party~$p_i$, we
can assign a probability for~$p_i$ to perform~$\act$ in \scn{} according to the probability space induced by \scn.
A central feature of probabilistic indistinguishability is captured as follows:

\begin{lemma}
\label{cor:action}
	Let \scn-A and \scn-B be two ensembles of algorithm \alg{} that are probabilistically indistiguishable to party~$p_i$.
	For each action~$\act$ of process $p_i$, the probability that~$p_i$ performs~$\act$ is equal in~\scn-A and \scn-B.
\end{lemma}

\begin{proof}
	Let \scn-A and \scn-B satisfy the assumptions, and let $\act$ be an action of process~$p_i$.
	Recall that the local state of process~$p_i$ contains the sequence of all actions that $p_i$ has performed up to its current state.
	Denote by $V_{\act}^{\text{A}}$ the set of nodes in \scn-A that
	represent configurations in which $\act$ appears only once in the
	sequence of local events contained in $p_i$'s local state, and it is the
	last element in that sequence.
	As before, for a node~$v\in\scn\text{-A}$ define by
	$S_v^{\text{A}}$ the set of executions in \scn-A that pass through~$v$.
	Recall that the set $S_v^{\text{A}}$ is an event in the probability space induced by \scn-A. 
	Thus,  $S_{\act}^{\text{A}}\triangleq\bigcup\limits_{v\in V_{\act}^{\text{A}}} S_v^{\text{A}}$ constitutes a measurable event in the probability space, and $P_{\text{A}}(S_{\act}^{\text{A}})$ corresponds to the probability that process~$p_i$ performs action~$\act$ in the ensemble \scn-A.
	Similarly, define $V_{\act}^{\text{B}}$, $S_{\act}^{\text{B}}$, and $P_{\text{B}}(S_{\act}^{\text{B}})$ with respect to \scn-B.
	It remains to show that $P_{\text{A}}(S_{\act}^{\text{A}}) = P_{\text{B}}(S_{\act}^{\text{B}})$.
	Recall that for a local state $\ls{i}$, we have defined the set of path sets $S_{\ls{i}}^{\scn}=\{S_v^{\scn}\mid v\in\scn{} \text{ s.t. $v$ represents a configuration that contains $\ls{i}$} \}$.
	
	Let $LV_{\act}^{\text{A}_i}$ be the set of nodes in the $i$-local ensemble  \scn-A$_i$ that are labeled with a local state~$\ls{i}$ in which $\act$ appears only once and is the last element in the sequence of local events.
	Similarly, define $LV_{\act}^{\text{B}_i}$ with respect to \scn-B$_i$.
	
	\begin{equation*}
		\begin{split}
			& P_{\text{A}}(S_{\act}^{\text{A}}) = 
			P_{\text{A}} \left(
			\bigcup\limits_{v\in V_{\act}^{\text{A}}} S_v^{\text{A}}
			\right) = 
			P_{\text{A}}\left( \bigcup\limits_{\lr{\ls{i},p}\in LV_{\act}^{\text{A}_i}} 
			\left(\bigcup\limits_{S_v^{\text{A}}\in S_{\ls{i}}^{\text{A}}} S_v^{\text{A}}\right)
			\right)
			= \sum\limits_{\lr{\ls{i},p_{\ls{i}}}\in LV_{\act}^{\text{A}_i}} P_{\text{A}}(\bigcup\limits_{S_v^{\text{A}}\in S_{\ls{i}}^{\text{A}}} S_v^{\text{A}}).
		\end{split}
	\end{equation*}
Where the last equality follows from the fact that for any $\lr{\ls{i}\!^1\!, p_{\ls{i}\!^1}} \ne \lr{\ls{i}\!^2\!, p_{\ls{i}\!^2}} \in LV_{\act}^{\text{A}_i}$ we have that $S_{\ls{i}\!^1}^{\text{A}}\bigcap S_{\ls{i}\!^2}^{\text{A}} = \emptyset$.
This is because $\act$ appears only once in the respective $\ls{i}$, so a path cannot contain two different local states in which $\act$ is both last and appears only once.
By definition of $\scn_i$, we have that $p_\ls{i} = P_{\scn}(\bigcup\limits_{S_v^\scn\in S_{\ls{i}}^\scn} S_v^\scn)$. Therefore,
\begin{equation*}
	\begin{split}
	\sum\limits_{\lr{\ls{i},p_{\ls{i}}}\in LV_{\act}^{\text{A}_i}} P_{\text{A}}(\bigcup\limits_{S_v^{\text{A}}\in S_{\ls{i}}^{\text{A}}} S_v^{\text{A}}) = 
	\sum\limits_{\lr{\ls{i},p_{\ls{i}}}\in LV_{\act}^{\text{A}_i}} p_{\ls{i}} .
	\end{split}
\end{equation*}

Similarly,
\begin{equation*}
	\begin{split}
		P_{\text{B}}(S_{\act}^{\text{B}}) =  
		\sum\limits_{\lr{\ls{i},p_{\ls{i}}}\in LV_{\act}^{\text{B}_i}} p_{\ls{i}} .
	\end{split}
\end{equation*}
	By definition of ensemble indistinguishability,  \scn-A$_i$= \scn-B$_i$, and so  $LV_{\act}^{\text{A}_i} = LV_{\act}^{\text{B}_i}$.
Hence, 
		$P_{\text{A}}(S_{\act}^{\text{A}}) =  P_{\text{B}}(S_{\act}^{\text{B}})$, as claimed.
\end{proof}

In the next section we  use ensembles and probabilistic indistinguishability to prove a lower bound on quality of decisions in Byzantine Agreement.
That is, on the probability of deciding on a bogus value suggested by dishonest parties.

\section{Byzantine Agreement with Qualitative Validity}
\label{sec:Validity}

%

Byzantine Agreement is one of the most fundamental problems in
distributed computing.
A set of $n$ parties, some of which might be Byzantine, 
need  to agree on the same value.
Ideally, we would like the decision to be on a value proposed by an
honest party.
And indeed, in the classic binary case~\cite{ben1983another}, where the
set of possible inputs is $ \{0, 1\}$, this is exactly what the
Validity property of Byzantine Agreement requires.
However, in the multi-valued case, in which inputs come from some
arbitrary domain $\mathbb{V}$, this is generally impossible to guarantee, because
one or more Byzantine parties can propose a value that is not proposed by honest
parties and otherwise act honestly~\cite{N94}.
Since multi-valued Byzantine Agreement protocols are the core of many
Blockchain systems~\cite{cachin2016architecture, Concord}, the issue of
preventing malicious attacks on the ``quality'' of decisions is
becoming more and more important.
The question is, therefore, what is the best validity property a
multi-valued Byzantine Agreement protocol can provide. 
That is, what are the conditions under which an algorithm can be guaranteed to decide
on a value proposed by an honest party and
what is the probability with which such a decision can be ensured
if these conditions fail to hold.

Two incomparable validity definitions, called weak Validity and external
	Validity, have been proposed for the
multi-valued case.
As in the binary case, \emph{weak Validity}~\cite{PSL80, correia2006consensus}
requires that if all honest parties propose the same value $v$, then
$v$ is the only value that can be decided.
However, if honest parties propose different values, then they can
decide on some pre-defined default value (which we denote by~$\bot$).
The initial motivation for weak validity was a spaceship cockpit with
four sensors, one of which might be broken~\cite{PSL80}.
However, from a contemporary practical point of view, such a definition is
useless for building Byzantine state machine replication (SMR) (e.g., as in
blockchains)~\cite{VABA, malkhi2019concurrency} since a decision of~$\bot$ 
in such a setting does not allow the system to make progress. Hence, in
order to guarantee progress, all honest parties must input the same
value (agree \textit{a priori}) even in failure-free runs.

To deal with this issue, Cachin, Kursawe, Petzold and
Shoup~\cite{CachinSecure} introduced the \emph{external Validity} property,
which allows the decision to be any value as long as it is valid
according to some external predicate (e.g., a valid transaction in a
blockchain system).
In particular, external validity does not preclude a situation in
which the decision value does not originate from an honest party.
To overcome this deficiency, Abraham, Malkhi, and
Spiegelman~\cite{VABA} extended the definition of external validity
with a \emph{decision quality} requirement, which bounds the
probability of the decision being a value proposed by the Byzantine
parties.
Specifically, they provide an algorithm that guarantees probability of
at least $1/2$ for the decision value to be an input of an honest
party.
Moreover, they claim in the paper that no algorithm provide a
stronger quality guarantee in the worst case scenario, but provide no
proof.
While their claim is very intuitive, it is not obvious how to prove it
without a notion such as probabilistic indistinguishability.

Note that the two variants of multi-valued validity are incomparable. On
the one hand, with external Validity parties never agree on $\bot$ and 
thus SMR progress is guaranteed by reaching a meaningful decision in
every slot.
On the other hand, honest parties may agree on a bogus value proposed by
malicious participants even if they agree \emph{a priori}. 
In addition, note that neither definition takes into account the
actual number of failures in the execution $f \leq t$. 
Below we define the \optimal{C} Validity property, which
promises progress and is stronger than each of these  validity conditions.

\subsection{Problem definition}
\label{sub:defintion}

In this section we assume a computationally bounded adversary that can corrupt up to $\threshold$ of the~$n$ parties, where $n=3t+1$.
Parties corrupted by the adversary are called
\emph{Byzantine} and may arbitrarily deviate from the protocol. 
Other parties are \emph{honest}.

Given an ensemble, we denote by $\newf \leq t$ the maximal number of parties the adversary corrupts
in any of the paths in the ensemble.
In addition, every party $p_i$ starts with an initial input value
$v_i$ from some domain $\mathbb{V}$, i.e., $p_i$'s local state in the
initial configuration (the root of the ensemble) is $s_i = v_i$.
We denote by $\Vin = \multiset{v_i \mid p_i \in \Pi}$ the multiset of all
input values.
For every multiset $\mathcal{M}$ and value $v \in \mathcal{M}$, we
denote by $\mult{v}{\mathcal{M}}$ the multiplicity of $v$ in $\mathcal{M}$.
The maximum multiplicity in a multiset $\mathcal{M}$ is denoted by
$\maxmult{\mathcal{M}} \triangleq max(\{\mult{v}{\mathcal{M}} \mid v
\in \mathcal{M}\})$.

We distinguish between static and adaptive adversaries and between weak
and strong ones.
A weak adversary does not observe the local states of  honest parties, whereas a
strong one does. 
A static adversary knows the input values but must determine the
corrupted parties at the start, i.e., immediately after the root.
An adaptive one is allowed to corrupt parties on the fly. 
To strengthen our result, we consider a weak and static adversary for
the lower bound, and a strong and adaptive one for the upper bound (the
algorithm).

The Agreement problem exposes an API by which a party can
\emph{propose} the input value and output a \emph{decision} from the
domain $\mathbb{V}$.
An Agreement algorithm is one that satisfies the Agreement, Termination
and Validity properties.
As deterministic solutions in failure-prone asynchronous systems are
impossible by FLP~\cite{FLP85}, we are interested in algorithms that
never compromise safety, and ensure liveness
almost surely. 
That is, we require that every ensemble~\scn{} of the algorithm
must satisfy the following properties:

\begin{itemize}
  
  \item Agreement: In every path (i.e., execution) of \scn{}, all
  honest parties that decide, output the same decision value.
  
  \item Probabilistic Termination: Every honest party decides with
  probability 1.
  
\end{itemize}

\noindent As for validity, we extend previous
definitions~\cite{PSL80,correia2006consensus,CachinSecure}
to capture the optimal conditions under which parties decide on a value
proposed by an honest party. 
Recall that input values are determined by the initial configuration,
which is at the root of an ensemble.
It follows that all executions in an ensemble share the same input vector.
The ensemble notation allows us to require
the following non-deterministic property:

\begin{itemize}   
 \item  \optimal{C} Validity: If $\maxmult{\Vin} -\newf \geq 2t+1$,
 then all honest parties that decide, output decision values in $\Vin$.
 Otherwise, the probability that they decide on a value
 in $\Vin$ is at least $1 - \frac{\newf}{n-t}$.
\end{itemize}

One important feature of \optimal{C} Validity is that in ensembles
without failures honest parties can only decide on a value in $\Vin$.
Moreover, parties never decide on a pre-defined $\bot$ and the
probability to decide on a value in $\Vin$ is proportional to the
number of Byzantine parties that actually occur in
the ensemble.
Many Validity definitions are phrased with relation to the input of correct processes. We, in contrast, consider all inputs and give the adversary a choice to corrupt the parties based on their input. Our phrasing does not weaken the Validity property (a direct proof appears in \Cref{???}).

\subsection{Tight bounds on Qualitative Validity}
\label{sec:lower}

In this section we prove that no algorithm for 
multi-valued Byzantine Agreement that satisfies Agreement and 
Probabilistic Termination can provide a better guarantee than
\optimal{C} Validity.
Moreover, we then show that this validity condition is the best we can achieve, 
by presenting an algorithm that satisfies this validity property. 
%

The following lemma states that if the condition in the \optimal{C}
Validity definition ($\maxmult{\Vin} -\newf \geq 2t+1$) does not hold,
then we can always find a multiset $\mathcal{M} \subset \Vin$ of size $n - t
- \newf$ such that no value in $\mathcal{M}$
has multiplicity higher than~$t$.

\begin{lemma}
\label{lem:subset}

Consider a multiset $\Vin$ of $n = 3t+1$ values and $0 < \newf \leq t$.\\
If $\maxmult{\Vin} -\newf < 2t+1$, then there is a multiset $\mathcal{M}
\subset \Vin$ such that $|\mathcal{M}| = n-t-\newf$ and
$\maxmult{\mathcal{M}} \leq t$.
\end{lemma}

\begin{proof}

%
%
%
%

Consider two cases:
\begin{itemize}
  
  \item  If $max\_mult(\mathcal{V}_{in}) \leq t$ then the claim's
  conclusion is true and the lemma holds.
  
  \item Otherwise, let $v \in \mathcal{V}_{in}$ such that $mult(v,
  \mathcal{V}_{in}) = max\_mult(\mathcal{V}_{in}) > t$.
Since, by assumption, $max\_mult(\mathcal{V}_{in}) - \newf < 2t+1$, we obtain $mult(v, \mathcal{V}_{in}) \leq 2t + \newf$.
This means that there are at least $n - (2t + \newf) = t - \newf + 1$
values distinct from $v$ in $\mathcal{V}_{in}$.
Define $\mathcal{M}$ to contain~$t$ copies of~$v$ along with $t -
\newf + 1$ values distinct from $v$. Then  
 $|\mathcal{M}| = 2t -\newf +1 = n - t-\newf$, and 
 $\maxmult{\mathcal{M}} \leq t$ as desired, since $\newf > 0$.
\end{itemize}
\vspace{-3mm}

\end{proof}

Clearly, if $\maxmult{\Vin} -\newf \geq 2t+1$,
then   \optimal{C} Validity guarantees that the decision value is
in~$\Vin$, which is the most that a validity property can ensure.
Therefore, for a validity property~$\Phi$ to be stronger
than \optimal{C} Validity, there must be $\fz$ and $\Vin$ for which
the probability to decide on a value in $\Vin$ according to~$\Phi$  is
strictly higher than $1 - \frac{\fz}{n-t}$.
However, for a validity property to be strictly stronger
than \optimal{C} Validity it must, in addition, satisfy \optimal{C}
Validity for all other values of $\newf \leq t$ and all other $\Vin$.
Formally, we say that a Validity property~$\Phi$
 is {\it strictly stronger} than \optimal{C} Validity if an algorithm
 \alg{} satisfying~$\Phi$ guarantees that:
\begin{enumerate}
\item For all $\newf\le t$ and $\Vin$ such that $\maxmult{\Vin} -\newf \ge
2t+1$, in every ensemble of \alg{} with $\newf$ corrupted parties
and the input multiset $\Vin$ an honest party that decides, decides on
a value in $\Vin$.
\item For all $\newf\le t$ and $\Vin$ such that $\maxmult{\Vin} -\newf < 2t+1$,
in every ensemble of \alg{} with $\newf$ corrupted parties and the
input multiset $\Vin$ the probability to decide on a value in $\Vin$ is at least $1 - \frac{\newf}{n-t}$.

\item There exist some $\fz\le t$ and $\Vin$ that satisfies $\maxmult{\Vin} -\fz
< 2t+1$, such that in every ensemble of \alg{} with $\fz$ corrupted
parties and the input multiset $\Vin$ the probability to decide on a value in
$\Vin$ is
strictly higher than $1 - \frac{\fz}{n-t}$.
\end{enumerate}

Intuitively, if $\maxmult{\Vin} -\newf < 2t+1$, then the adversary can
delay $t$ honest parties until the decision is made such that the
remaining honest parties have input values in a multiset
$\mathcal{M}$ that satisfies the property in Lemma~\ref{lem:subset}.
Therefore, we say that a ``fair share'' probability for a value $v \in
\mathcal{M}$ to be decided is proportional to $\mult{v}{\mathcal{M}}$
and equal to $\frac{\mult{v}{\mathcal{M}}}{n-t}$.
Roughly speaking, by (3), for a validity condition to be strictly
stronger than \optimal{C} Validity, there must be a value in $v \in \mathcal{M}$ that
is decided with a probability that is higher than its ``fair share''.  
We show that in this case there is a probabilistically
indistinguishable ensemble in which $v$ is proposed only by corrupted
parties and the probability to decide on $v$ is the same.
As a result, in that ensemble the corrupted parties get more than their ``fair share'',
and thus violate (2).
Formally, to show that \optimal{C} Validity is optimal
we prove the following:

\begin{theorem}
	\label{thm:lower}
No asynchronous Byzantine Agreement algorithm  satisfies a
validity property~$\Phi$ that is strictly stronger than \optimal{C} Validity even against a weak and static adversary. 
\end{theorem}

\begin{proof}

Assume, by way of contradiction, that there is such an algorithm~\alg.
In particular, there exist $\fz\le t$ and $\Vin$ that satisfy
$\maxmult{\Vin} -\fz< 2t+1$, and in every ensemble of \alg{} with $\fz$
(maximal corrupted parties in any path in the ensemble) and the input
multiset $\Vin$ the probability to decide on a value in $\Vin$ is strictly higher than $1 -
\frac{\fz}{n-t}$.
To show the contradiction we use $\Vin$ to construct an ensemble $\ens$
of \alg{} with a multiset $\Vin' \neq \Vin$ of input values and
$\newf' \le t$ corrupted parties such that the probability to agree on a
value in $\Vin'$ in $\ens$ is strictly lower than
$1-\frac{\newf'}{n-t}$. This contradicts either condition (2) or (3) from the properties that $\Phi$ must satisfy in order to be strictly
stronger than \optimal{C} Validity.
We next describe a few ensembles under different adversarial strategies
and use Probabilistic Indistinguishability between them to prove the
theorem.

\begin{ensemble}
\label{exec1}

$\Vin^1 = \Vin$ and $\newf^1 = \fz$.
By Lemma~\ref{lem:subset}, there is a multiset $\mathcal{M} \subset \Vin^1$
such that $|\mathcal{M}| = n-t-\fz$ and $\maxmult{\mathcal{M}} \leq t$.
Let $M \subset \Pi$ be the set of parties with inputs in $\mathcal{M}$, i.e.,
the multiset $\mathcal{M}=\{v_i \mid p_i \in M\}$.
Let $F, T \subset \Pi$ be two sets of parties such that $|F| = \fz$, 
$|T| = t$, and $ \Vin \setminus \mathcal{M} =\{v_i \mid p_i \in F\cup T
\} $.
The adversary (statically) corrupts the parties in $F$ and delays all
messages from parties in $T$ until all honest parties in $M$
decide.
Messages sent among parties in $M \cup F$ are immediately delivered. 
Let $v_f \not\in V^1_{in}$. 
The corrupted parties act like honest parties that get $v_f$ as an
input.

\end{ensemble}

\begin{ensemble}
\label{exec2}

Consider a multiset of input values $\Vin^2$ that is identical to $\Vin^1$
except that parties in $F$ get $v_f$ as an input.
The adversary corrupts parties in $T$.
Messages sent among parties in $F \cup M$ are immediately delivered.
Finally the corrupted parties send no messages.

\end{ensemble}

\begin{figure*}[th]
    \centering
    \begin{subfigure}[t]{0.47\textwidth}
        \centering
        \includegraphics[height=1.2in]{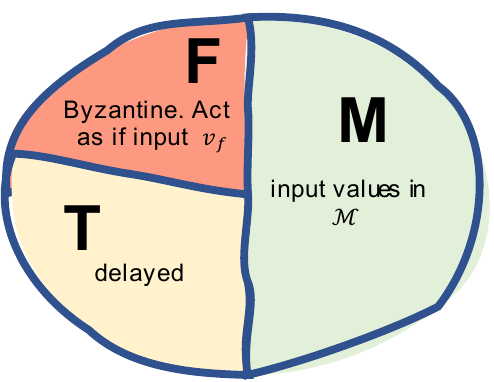}
        \caption{Ensemble~\ref{exec1}. Parties in $F$ are Byzantine that
        act as if they are honest with input $v_f \not\in \Vin^1$.
    Messages from parties in $T$ are delayed.}
    \end{subfigure}\hfill%
	\hfill
    \begin{subfigure}[t]{0.47\textwidth}
        \centering
        \includegraphics[height=1.2in]{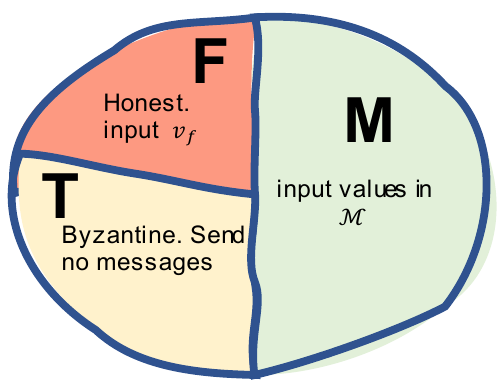}
        \caption{Ensemble~\ref{exec2}. Parties in $F$ are honest and
        $v_f \in \Vin^2$. Parties in $T$ are Byzantine that send no
        messages.}
    \end{subfigure}
    \caption{Sets $F, T$ and $M$ divide the
        parties into 3 disjoint sets. In both ensembles the set $M$
        contains parties with inputs from $\mathcal{M}$. Ensemble~\ref{exec1} and \cref{exec2} are probabilistically indistinguishable to all
parties in~$M\!$.}
    \label{fig:ensemble2}
\end{figure*}

Clearly, by \Cref{def:PI}, \cref{exec1}
and \cref{exec2} are probabilistically indistinguishable to all processes
in~$M$ (see illustration in Figure~\ref{fig:ensemble2}).
By the Probabilistic Termination property, all parties in $M$ decide in
Ensemble~\ref{exec2} with probability~1.
Therefore, \Cref{cor:action} implies that, in Ensemble~\ref{exec1}, all
parties in $M$  decide with probability 1 as well.

For every value $v \in \Vin$ denote by $P^1(v)$ the probability that
the honest parties in $M$ decide on~$v$ in Ensemble~\ref{exec1}.
By  assumption, the probability to decide on a value
in $\Vin$ is strictly higher than $1- \frac{\fz}{n-t}$ in
Ensemble~\ref{exec1}, i.e., \mbox{$\sum_{v \in \Vin} P^1(v) > 1-
\frac{\fz}{n-t}$}.
Therefore, there are two cases for Ensemble~\ref{exec1}:
\begin{itemize}
  
  \item First, there is a value $u \in \mathcal{M}$ s.t.\
  $P^1(u) > \frac{\mult{u}{\mathcal{M}}}{n-t}$.
  
  
   \item Otherwise, there is a value $w \not\in \mathcal{M} \cup \{v_f\}$ 
   such that $P^1(w) > 0$.
  
\end{itemize}

\noindent \textbf{First case:}
For the first case, let $U = \{p_i \in M \mid v_i = u \}$.
Since $\maxmult{\mathcal{M}} \leq t$, we get $|U| \leq t$.  
Consider the following ensemble:

\begin{ensemble}
\label{exec3}
The multiset of input values $\Vin^3$ is identical to
$\Vin^1$ except (1) parties in
$F$ get $v_f$ as an input; and all other parties that get $u$ as an
input in Ensemble~\ref{exec1} get some value $v' \neq u$.
(Note that $u \not\in \Vin^3$).
The adversary corrupts all parties in $U$ s.t.\ corrupted parties act as
honest parties that got $u$ as an input. $\newf^3 = |U|$.
Messages from parties in $T$ are delayed until all parties in $M
\setminus U$ decide, and messages sent among parties in $M \cup F$ are
immediately delivered.
\end{ensemble}

\begin{figure*}[th]
    \centering
    \begin{subfigure}[t]{0.47\textwidth}
        \centering
        \includegraphics[height=1.2in]{scenario1.pdf}
        \caption{Ensemble~\ref{exec1}. Parties in $F$ are Byzantine that
        act as if they honest with input $v_f \not\in \Vin^1$.
    Messages from parties in $T$ are delayed.}
        \end{subfigure}\hfill%
    \hfill
    \begin{subfigure}[t]{0.47\textwidth}
        \centering
        \includegraphics[height=1.2in]{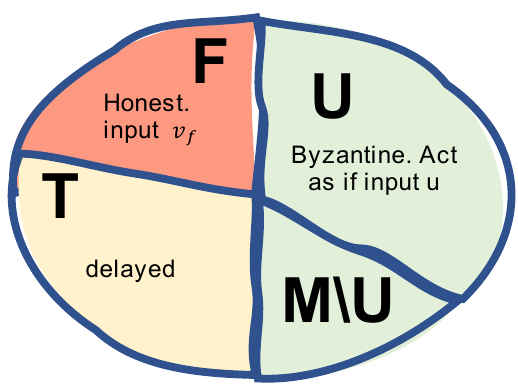}
        \caption{Ensemble~\ref{exec3}. Parties in $F$ are honest and
        $v_f \in \Vin^3$. Messages from parties in $T$ are delayed.
        The set $U \subset M$ contains parties that input $u$
        in Ensemble~\ref{exec1}. Here they are Byzantine that
        act as if they honest with input $u$. The value $u$ is not in
        $\Vin^3$.}
    \end{subfigure}
    \caption{In both ensembles parties in $M\setminus U$
        input the same values. Ensemble~\ref{exec1} and
        \Cref{exec3} are probabilistically indistinguishable to all
        processes in~$M \setminus U$.}
    \label{fig:ensemble3}
\end{figure*}

Note that Ensemble~\ref{exec1} and Ensemble~\ref{exec3} are
probabilistically indistinguishable for parties in $M \setminus U$ (see
illustration in Figure~\ref{fig:ensemble3}).
Therefore, by \Cref{cor:action}, the probability of parties in $M
\setminus U$ to agree on $u$ in Ensemble~\ref{exec3} is higher than
$\frac{|U|}{n-t}$.
Consequently, the probability to agree on a value in
$V^3_{in}$ in Ensemble~\ref{exec3} is strictly lower than $1 -
\frac{|U|}{n-t}$ = $1 - \frac{\newf^3}{n-t}$.
This contradicts the assumption that \alg{} satisfies a strictly stronger validity
property than \optimal{C} Validity.

\noindent \textbf{Second case:}
In this case there is a value $w \not\in \mathcal{M} \cup \{v_f\}$ such that $P^1(w) > 0$. 
Consider the following ensemble:

\begin{ensemble}
\label{exec4}

Consider $\Vin^4$ to be the multiset of input values that is
identical to $\Vin^1$ except parties in $F$ get $v_f$ as an input and
all other parties that are assigned~$w$ from~$\Vin^1$ as an input, are assigned instead some $v\ne w$ ($w \not\in V^4_{in}$).
The adversary corrupts no parties ($\newf^4 =0$).
Messages from parties in $T$ are delayed until all parties in $M$
decide, and messages sent among parties in $M \cup F$ are 
immediately delivered.

\end{ensemble}

\begin{figure*}[th]
    \centering
    \begin{subfigure}[t]{0.47\textwidth}
        \centering
        \includegraphics[height=1.2in]{scenario1.pdf}
        \caption{Ensemble~\ref{exec1}. Parties in $F$ are Byzantine that
        act as if they honest with input $v_f \not\in \Vin^1$.
    Messages from parties in $T$ are delayed.}
        \end{subfigure}\hfill%
    \hfill
    \begin{subfigure}[t]{0.47\textwidth}
        \centering
        \includegraphics[height=1.2in]{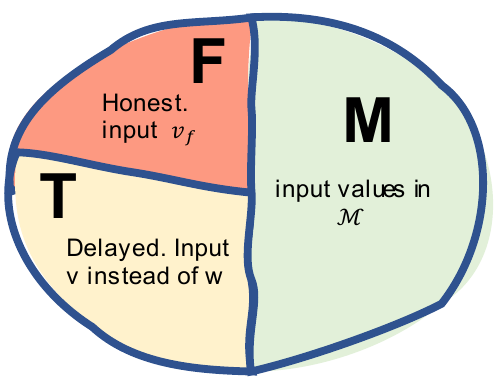}
        \caption{Ensemble~\ref{exec4}. All parties are honest.
        Messages from parties in $T$ are delayed.
        The value $v_f \in \Vin^4$, whereas the value $w \not\in
        \Vin^4.$}
    \end{subfigure}
    \caption{Ensemble~\ref{exec1}
and \Cref{exec4} are probabilistically indistinguishable to all
processes in~$M\!$.}
    \label{fig:ensemble4}
\end{figure*}

Note that \Cref{exec1} and \Cref{exec4} are
probabilistically indistinguishable for parties in $M$ (see illustration
in Figure~\ref{fig:ensemble4}).
Therefore, by \Cref{cor:action}, the probability of a party in
$M$ to agree on $w$ in \Cref{exec4} is higher than $0$.
Thus, since $w \not\in V^4_{in}$, the probability to agree on a value
in $\Vin^4$ in \Cref{exec4} is strictly lower than $1$. 
Since  $\newf^4=0$ in \Cref{exec4}, we get that the
probability to agree on a value in $\Vin^4$ in \Cref{exec4} is
strictly lower than $1 - \frac{\newf^4}{n-t}$. 
This contradicts the assumption that \alg{} satisfies 
a validity
property $\Phi$ that is strictly stronger than \optimal{C} Validity, completing the proof.
\end{proof}

The lower bound result of \Cref{thm:lower} is tight, as the following theorem shows: 

\begin{theorem}
	\label{thm:upper}
	There exists an asynchronous Byzantine Agreement algorithm that satisfies  \optimal{C} Validity against a strong and adaptive adversary.
\end{theorem}

In fact, this theorem shows a bit more than that the bound is tight. 
While \Cref{thm:lower} showed that no better than Qualitative Validity can be achieved  even against a weak and static adversary, \Cref{thm:upper} shows that it {\it is} achievable, and this can be done against a much stronger adversary.

We defer the proof of \Cref{thm:upper} to Appendix~\ref{app:upper}, where we show  that a sequential composition of
two known algorithms, from~\cite{raynal2017signature} and~\cite{VABA},
yields the first Agreement protocol that satisfies the
Agreement, Probabilistic Termination, and \optimal{C} Validity properties.
In a nutshell, we first run the algorithm
in~\cite{raynal2017signature}.
By its weak validity property, if all honest parties start with the same value,
then they all decide on it.
Otherwise, they decide on $\bot$, in which case we invoke the algorithm
in~\cite{VABA} and output its decision value.
The resulting asynchronous algorithm achieves \optimal{C} Validity.
Moreover, it does so in expected
constant number of rounds, using $O(n^2)$ communication
complexity, and is resilient against $t < n/3$ Byzantine parties.
Each of these parameters is known to be optimal in this
setting~\cite{FLP85,VABA,bracha1987asynchronous}.

\section{Discussion}
\label{sec:discussion}

Validity is one of the essential properties that agreement algorithms are required to satisfy. For multi-valued agreement, several distinct versions of validity have been studied in the literature over the years. 
Indeed, the desire to provide a variety of quality and fairness features in the blockchain world has given rise to new validity properties that are especially suited to randomized agreement algorithms. 
In this work we have introduced a new, probabilistic, validity property for multi-valued agreement. 
Called {\it \optimal{C} Validity}, this notion is strictly stronger than two  popular validity conditions, which  are incomparable to one another. Intuitively, it  bounds the probability that corrupted parties will cause the algorithm to decide on a bogus  value. 
Our main theorem is that, in a precise sense, \optimal{C} Validity is the strongest validity property that can be satisfied by an asynchronous Byzantine Agreement algorithm.

In order to prove our lower bound,  we represented adversary strategies in terms of a mathematical object called an ensemble, and  introduced the notion of probabilistic indistinguishability between ensembles.  
This framework facilitates the statement of probabilistic properties of algorithms, and the proof of lower bounds on such properties. 
Our framework is applicable beyond the proof of our theorem.  For example, as discussed in  \Cref{sec:related}, the lower bound proof  by Attiya and Censor-Hillel in~\cite{attiya2010lower}  constructs multiple  (deterministic) indistinguishability chains to account for different random tapes. Their construction can be replaced by a single 
probabilistic indistinguishability chain among ensembles.  
Another place where our framework fits seamlessly is the probabilistic lower bound theorem of~\cite{abraham2019communication}. Their technical argument shows that two adversaries are indistinguishable to a particular party. While they do not provide a formal definition of indistinguishability between adversaries, interpreting their proof using our definitions fills this gap perfectly. 
Moreover, as we show in \Cref{sec:formalize}, the added transparency into their proof that is obtained by couching it using our framework allows us to strengthen their claim: Their probabilistic lower bound holds for a strictly weaker adversary than is claimed in their theorem. 

Probabilistic indistinguishability captures the fact that a given party will perform the same actions, with the same probability, in both ensembles. Indeed, our theorem provides a tight bound on the probability that an honest party will decide on a good value. 
However, in the design and analysis of randomized distributed algorithms one may be interested in the correlation between actions of several parties. E.g., 
whether parties decide on the same value (i.e., satisfy Agreement) with a sufficiently high probability. 
The notion of probabilistic indistinguishability does not account for 
such correlation among actions. 
An interesting topic for future investigation is to formulate notions that will account for the correlation among actions of different parties. We expect that such notions will be generalizations of our definition of probabilistic indistinguishability. Their precise form is left as an open problem.

%
%
%
%
%
%
%
%
%
%





\bibliography{bibliography}

\newpage
\appendix

\section{Proving Theorem 4 from \cite{abraham2019communication} (ACDNPLS)}
\label{sec:formalize}
In \emph{(binary) Byzantine Broadcast}, an \textit{a priori} fixed designated \emph{sender} starts out with an input bit $b\in \{0,1\}$.
An algorithm \alg{} solves binary Byzantine Broadcast with probability at least~$x$ if, in every ensemble of \alg, with the probability of at least~$x$ all of the three following properties hold. 
\begin{itemize}
	\item Agreement. All honest parties that output a bit, output the same bit.
	
	\item Termination. Every honest party outputs a bit.

	\item Validity. If the sender is honest and the sender’s input is $b$, then all	honest parties output $b$.
\end{itemize}

The following theorem from~\cite{abraham2019communication} considers a model with a non-uniform p.p.t. strongly adaptive adversary.
That is, the adversary can (1) leverage randomness in its favor, \ and (2) observe that a message is sent at time $\tp$ by any party $p_i$ to any other party, decide to adaptively corrupt $p_i$, and remove the messages sent by $p_i$ at time $\tp$.

\begin{theorem} [ACDNPLS \cite{abraham2019communication}]
	If a protocol solves Byzantine Broadcast with $\frac{3}{4} + \epsilon$ probability against a non-uniform p.p.t. strongly adaptive adversary, then in expectation, honest nodes collectively need to send at least $(\epsilon t)^2$ messages.
\end{theorem}

Using our definitions from this paper, we can now better formalize their statement (and slightly strengthen it).
Our proof follows the outline of the proof in~\cite{abraham2019communication}.
 \\[\parskip]

\noindent
\textbf{Theorem 2$'$.}
\textit{If an algorithm \alg{} solves (in a model with a strongly adaptive polynomial time adversary) Byzantine Broadcast with probability at least $\frac{3}{4} + \epsilon$, then there exists an ensemble of \alg{} in which honest parties collectively send at least $(\epsilon t)^2$ messages in expectation.}


\begin{proof}
	Let \alg{} be an algorithm that solves Byzantine Broadcast with probability at least $\frac{3}{4} + \epsilon$.
	Assume by way of contradiction that in \textbf{every ensemble} of \alg, the honest parties collectively send fewer than $(\epsilon t)^2$ messages in expectation.

	Without loss of generality, assume that there exist $\lceil n/2 \rceil$ parties each of which outputs 0 with probability at most 1/2 if they receive no messages. (Otherwise, then there must exist $\lceil n/2 \rceil$ nodes that output 1 with at most 1/2 probability if they receive no messages, and the entire proof follows from a symmetric argument.)
	Formally, the set $N_0\subset \Pi$ is of size $\lceil n/2 \rceil$, and for every $p_i\in N_0$ it holds that in every ensemble \scn{} of \alg{}, $P_\scn[p_i \text{ outputs 0}\mid p_i \text{ receives no messages in the path}]\le 1/2$. 
	Let $F\subset N_0$ be a set of $t/2$ these parties, not containing the designated sender. Note that these nodes may output 1 or they may simply not terminate if they receive no messages. (We can always find such an $F$ because $t/2 < \lceil n/2 \rceil$.)
	Consider the following ensemble \scn-A. \\[\parskip]

\noindent
\textbf{Ensemble A.}
		\textit{The sender's input bit is 0, all messages that are sent, are synchronously delivered, and the parties in $F$ are corrupted.
		Specifically, a party in $F$ behaves honestly (according to its protocol) except for \ (a) not sending messages to any other party in $F\!$, and \ (b) ignoring the first $t/2$ messages it receives from parties in $\Pi\setminus F$.}\\[\parskip]

	Let $z_{\text{A}}$ be a random variable,  in \scn-A, denoting the number of messages sent by parties in $\Pi\setminus F$ to~$F$. 
	By the assumption we have that $\mathbb{E}[z_{\text{A}}] < (\epsilon t)^2$.
	Let $X_1$ be the event that $z_{\text{A}} \le \frac{(\epsilon t)^2}{2\epsilon}$.
	By Markov’s inequality, $P_{\scn\text{-A}}[z_{\text{A}}> \frac{1}{2\epsilon} \mathbb{E}[z_{\text{A}}]] < 2\epsilon$.
	Thus, the probability of the event $X_1$ in \scn-A is $P_{\scn\text{-A}}[z_{\text{A}} \le \frac{\epsilon t^2}{2}] \ge P_{\scn\text{-A}}[z_{\text{A}} \le \frac{1}{2\epsilon} \mathbb{E}[z_{\text{A}}]] > 1 - 2\epsilon$.
	
	Let $p_f$ be the party in $F$ with the highest probability of receiving at most $t/2$ from the first $\frac{\epsilon t^2}{2}$ messages sent by honest parties in \scn-A.
	Formally, denote by $x_i$ the random variable, in \scn-A, corresponding to the number of messages received by $p_i$ out of the first $\frac{\epsilon t^2}{2}$ messages sent by honest parties.
	Then $p_f \triangleq \arg\max\limits_{p_i\in F} P_{\scn\text{-A}}[x_i \le t/2]$.
	Notice that in each path in \scn-A there are a total of at most $\frac{\epsilon t^2}{2}$ messages to distribute among the $t/2$ parties in $F$.
	Therefore, fewer than $\epsilon t$ parties in $F$ receive more than $t/2$ messages, and at least $|F|-\epsilon t = (1/2 - \epsilon) t$ parties receive at most $t/2$ of the mentioned messages.
	This implies that the expected probability of a party in $F$ to receive at most $t/2$ of the messages is at least $\frac{|F|-\epsilon t}{|F|} = 1 - 2\epsilon$.
	Let $X_{p_f}$ be the event in \scn-A that $x_f \le t/2$.
	By choosing $p_f$ in the way we did, we obtain that $P_{\scn\text{-A}}[X_{p_f}]\ge 1 - 2\epsilon$.
	
	In \scn-A, the probability that at most $\frac{(\epsilon t)^2}{2\epsilon}$ messages from honest parties are sent to $F$, whilst party $p_f$ receives at most $t/2$ of those messages is:
	$$
	P_{\scn\text{-A}}[X_1 \cap X_{p_f}] = P_{\scn\text{-A}}[X_1] + P_{\scn\text{-A}}[X_{p_f}] - P_{\scn\text{-A}}[X_1 \cup X_{p_f}] > (1-2\epsilon) + (1-2\epsilon) - 1 = 1 -4\epsilon.
	$$
	
	Now consider the following ensemble \scn-B, which is very similar to \scn-A. \\[\parskip]

\noindent
\textbf{Ensemble B.}	
		\textit{The sender's input bit is 0, and parties in $F\setminus \{p_f\}$ are corrupted and behave exactly as in \scn-A.
		In addition, parties in $\Pi\setminus F$ behave as in \scn-A (according to the algorithm~\alg) except that the first $t/2$ messages that are supposed to be sent by $\Pi\setminus F$ to $p_f$ are now omitted from~\scn-B.
		In order to do so, at most $t/2$ parties in $\Pi\setminus F$ are also corrupted (in an adaptive manner) in \scn-B.
		Other than this, the corrupted parties in $\Pi\setminus F$ behave exactly as in \scn-A (including sending later messages to $p_f$). } \\[\parskip]

We observe that for a party $p_i\in \Pi\setminus F$ that is honest in \scn-B, the ensembles \scn-A and \scn-B are probabilistically indistinguishable.
Therefore, by \Cref{cor:action}, their protocols prescribe the same (possibly probabilistic) behavior in both of the ensembles.
The fact that all honest parties act according to \alg, together with the fact that the maximal number of corrupted parties in \scn-B at most $|F|-1+t/2 = t-1 \le t$, mean that \scn-B is an ensemble of~\alg.

By construction, $p_f$ receives no messages in \scn-B with the same probability as $P_{\scn\text{-A}}[X_1 \cap X_{p_f}] = 1 -4\epsilon$.
Recall that (by the definition of $F$), party $p_f$ outputs 0 with probability at most $1/2$ on the set of paths in which it receives no messages whatsoever.
Let $Y_f$ be the complementary event, in \scn-B, to $p_f$ outputting 0. That is, $Y_f$ is the event that $p_f$ outputs 1 or does not terminate at all.
Then $P_{\scn\text{-B}}[Y_f] > \frac{1}{2}(1 -4\epsilon)$.

Moreover, since \scn-A and \scn-B are probabilistically indistinguishable to honest parties in $\Pi\setminus F$, by \Cref{cor:action}, the probability of these parties to output 0 is equal in \scn-A and \scn-B, and is at least $\frac{3}{4} + \epsilon$ by assumption.
Let $Y_0$ be the event in \scn-B that all honest parties in $\Pi\setminus F$ output 0.
If both $Y_0$ and $Y_f$ occur, then this path (execution) in \scn-B violates either Agreement of Termination.
We get 
$$
P_{\scn\text{-B}}[Y_0 \cap Y_f] = P_{\scn\text{-B}}[Y_0] + P_{\scn\text{-B}}[Y_f] - P_{\scn\text{-B}}[Y_0 \cup Y_f] > (\frac{3}{4} + \epsilon) + \frac{1}{2}(1 -4\epsilon) - 1 = \frac{1}{4} - \epsilon.
$$
This contradicts the assumption that \alg{} solves Byzantine Broadcast with probability at least \mbox{$\frac{3}{4} + \epsilon$.}
\end{proof}

This proof followed the outline of the original one and made use of the same underlying structure. 
But the use of ensembles and probabilistic indistinguishability provided the arguments with a rigorous foundation and removed the ambiguity in their use of indistinguishability.
The improved transparency also made it obvious that enabling the adversary to employ randomization is unnecessary. Hence, the result could have been slightly strengthened by weakening the adversary's abilities.

\section{Pseudo-code for Construction of a local ensemble}
\label{sec:pseu...}

We construct $\scn_i$ in two reiterating steps:
(1) We connect a node directly to its closest descendants in which $i$'s local state does not change (and remove all nodes in between).
(2) For each node -- starting from the root and going down the tree -- we combine sons that represent the same local state by merging their subtrees and assigning its root (the combined node) the sum of the probabilities.
From a complexity perspective this equals to two simultaneous BFS runs on the (possibly infinite) $\ens$ tree.
\begin{algorithm}
	\caption{Building local ensemble from ensemble \ens~ for party $p_i$}
	
	\begin{algorithmic}[1]
	
	\Statex $\emph{NodesQueue}$ - An initially empty queue. 
	\Statex $root$ is the root of \ens
	\Statex 
	
	\State create node $v_{root_i}$
	\State $\ls{i}^{root} \gets \emph{root.configuration}[i]$
	\State $Start_{root_i} \gets \{\langle root,1 \rangle\}$
	\State $\emph{NodesQueue}.enqueue(\langle v_{root_i}, \ls{i}^{root}, Start_{root_i} \rangle)$
	\While{\emph{NodesQueue} not empty}
	
	   \State $\langle v_{\ls{i}}, \ls{i}, Start_{\ls{i}} \rangle \gets \emph{NodesQueue}.dequeue()$
	   \State $\emph{LabelNode}(v_{\ls{i}},\ls{i}, Start_\ls{i})$
	
	\EndWhile 
	
	\Statex
	
	\Procedure{LabelNode}{$v_{\ls{i}},\ls{i}, Start_\ls{i}$}
	   
	   \State $\emph{SetOfChildren} \gets \{\}$
	
	    \State $p_\ls{i}\gets \sum\limits_{v \in Start_\ls{i}} p_v$

        \For{\textbf{each} $v\in Start_\ls{i}$}
             \For{\textbf{each} $u\in\Child(v)$}
                \State propagationQueue.enqueue($\lr{u,p_v\cdot\text{weight}(v,u)}$)
             \EndFor
        \EndFor   
        
        \While{propagationQueue not empty}
        
        \State $\lr{u,p_u}\gets$ propagationQueue.dequeue()
        \If{$u.$configuration$[i]=\ls{i}$} 
            \For{\textbf{each} $w\in\Child(u)$}
            \State propagationQueue.enqueue( $\lr{w, {p_u\cdot\text{weight}(u,w)} }$ )
            \EndFor
        \Else \Comment{$u.$configuration$[i]\ne\ls{i}$}
            \State $\emph{SetOfChildren} \gets \emph{SetOfChildren} \cup \{u.\text{configuration}[i]\}$
            \If{the set $Start_{u.\text{configuration}[i]}$ does not exist} 
                \State  $Start_{u.\text{configuration}[i]} \gets \{\}$
            \EndIf  
            \State $Start_{u.\text{configuration}[i]} \gets Start_{u.\text{configuration}[i]} \cup \{\lr{u,p_u}\} $
             
        \EndIf 
        \EndWhile
        \State label node $v_{\ls{i}}$ with $\lr{\ls{i},p_\ls{i}}$
        
        
        \For{\textbf{each} $\ls{i}^j\in \emph{SetOfChildren}$} 
            \State create node $v_{\ls{i}^j}$ and add it to $\Child(v_\ls{i})$
            \State $\emph{NodesQueue}.enqueue(\langle v_{\ls{i}^j}, \ls{i}^j, Start_{\ls{i}^j} \rangle)$
        \EndFor

	\EndProcedure
	
	\Statex

%
%
%

		
%
%
	\end{algorithmic}
	\label{alg:try1}
\end{algorithm}

%
%


\section{Consensus with Qualitative Validity is Solvable}
\label{app:upper}

In this section we show that a sequential composition of two
known algorithms, from~\cite{raynal2017signature} and~\cite{VABA},
yields an Agreement protocol that satisfies the Agreement,
Probabilistic Termination, and \optimal{C} Validity properties.
We next overview the properties guaranteed by each of the algorithms,
then show how to combine them to achieve \optimal{C} Validity, and
finally prove correctness and analyse complexity.

\paragraph*{Overview of the RM protocol's properties.} 
The asynchronous Agreement protocol
proposed by Raynal and Mostefaoui (RM)~\cite{raynal2017signature}
satisfies Agreement, Probabilistic termination, Weak validity, and Non-intrusion against an adaptive adversary.
The Weak validity property is a variant of the first part of
the \optimal{C} Validity.
That is, if all input values are the same
($\maxmult{\Vin} = 3t+1$), then honest parties can only decide on
this value.
However, if parties start with different values ($\maxmult{\Vin} < 3t+1$),
then they are allowed to agree on a pre-defined $\bot$ value.
The Non-intrusion property requires that honest parties decides on
values in $\Vin \cup \{\bot\}$.
That is, honest parties never decide on a value promoted by the
adversary.

The complexity of RM~\cite{raynal2017signature} is
the following: the protocol (1) tolerates up to $t< n/3$ Byzantine
parties, (2) runs in expected constant number of rounds, and (3) sends
$O(n^2)$ words in $O(n^2)$ messages where a word contains a constant
number of signatures and values. 

\paragraph*{Overview of the AMS protocol's properties.}
The asynchronous agreement protocol proposed by Avraham,
Malkhi, and Spiegelman (AMS)~\cite{VABA} satisfies the Agreement,
Probabilistic termination, External validity, and Quality properties against an
adaptive adversary.
The External validity property requires that honest parties decide
on values that are valid by some external predicate.
The Quality property requires that the probability that the decision value is in
$\Vin$ is at least $1/2$.
The complexity of AMS is similar to that of RM.

\paragraph*{Sequential composition}
We show that a sequential composition \alg$^s$  of the
RM and AMS protocols satisfies Agreement, Probabilistic termination, and \optimal{C} Validity with an optimal resilience and complexity in an asynchronous
setting with an adaptive adversary.

The pseudocode of \alg$^s$  appears as Algorithm~\ref{alg:optimal}.
First, parties try to reach agreement via the
RM protocol and if its decision
value is not $\bot$, then the parties decide on this value.
Otherwise, they propose their input value in the AMS protocol and decide on its decision value.
Although the Weak validity property that is proved for RM in~\cite{raynal2017signature} does not imply the first part of \optimal{C} Validity, a small modification of the proof proves that the protocol indeed satisfies it.
Moreover, the non-intrusion property guarantees that even if
$\maxmult{\Vin} -f < 2t+1$ then, in ensembles
of \alg$^s$ with $f$ maximal corrupted parties and an input multiset $\Vin$, parties never decide in
line~\ref{line:decide1} of \Cref{alg:optimal} on a value that is not in
$\Vin$.
To prove that the protocol in \Cref{alg:optimal} satisfies \optimal{} validity, we need to
show that the AMS protocol satisfies the second part of
\optimal{C} Validity. 
That is, in every ensemble of \alg$^s$ with $f$ maximal corrupted parties and an input multiset $\Vin$,
 the probability to decide on a value in $\Vin$ is at least $1 - \frac{f}{n-t}$.
By the quality property of AMS we get a probability of at
least $1/2$ in the worst case when $f = t$. 
Below we overview the main part of the AMS algorithm and
prove that the protocol indeed satisfies a stronger property, i.e.,
the second part of \optimal{C} Validity.

\begin{algorithm}
\caption{An agreement algorithm with Qualitative validity: protocol for
party $p_i$}

\begin{algorithmic}[1]

	\State $y_i \gets $\emph{RM-propose}$(v_i)$
	\If{$y_i \neq \bot$}

		\State decide $y_i$
		\label{line:decide1}

	\Else

		\State decide \emph{AMS-propose}$(v_i)$

	\EndIf

\end{algorithmic}
\label{alg:optimal}
\end{algorithm}

\subsection{Analysis}

The Agreement and Probabilistic-termination properties of the
composition follows immediately from the ones in
RM and AMS since $y_i\neq \bot$ in line~2 is true or false for all
correct processes categorically.
As for \optimal{C} Validity, it follows from the weak validity proof in
RM that if $\maxmult{\Vin} -f \geq 2t+1$, then all honest parties that
decide, decide on values in $\Vin$.
Moreover, by the non-intrusion property of RM, honest parties can
only decide in line~\ref{line:decide1} on a value in $\Vin$.
Thus, to prove \optimal{C} Validity we need to show that AMS satisfies
that the probability to decide on a value in $\Vin$ is at least $1 -
\frac{f}{n-t}$.
Below we overview the relevant parts of AMS and then
prove that it indeed satisfies the required property.

In a nutshell, in every round of the AMS protocol,
parties concurrently promote their values via some broadcast algorithm
until at least $n-t$ broadcasts complete. 
Then, using a shared global coin, parties elect one broadcast instance
uniformly at random and ignore the rest.
If a completed broadcast is elected, then its value is \emph{fixed}
and parties will eventually decide on this value.
Otherwise, parties continue to the next round with either their value
from the previous round or with the value of the elected broadcast in
this round.
Given the above description we prove the following lemma:

\begin{lemma}

The probability to decide on a value proposed by an honest party in
the AMS protocol is at least $1 - \frac{f}{n-t}$.

\end{lemma}

\begin{proof}

We prove the lemma by showing that the probability to decide on a
value proposed only by Byzantine parties is at most $\frac{f}{n-t}$.
To bound this probability from above, we assume that if a broadcast by
a Byzantine sender is elected in some round $r$, then all honest
parties decide on its value in round $r$.
That is, we assume that all Byzantine parties complete their
broadcasts in all rounds before the parties randomly elect one broadcast
instance.
By the AMS protocol, for every round $r$, if a not completed 
broadcast is elected, then honest parties continue to the next round
with either their values from the previous round or with the value of
the elected broadcast.
Therefore, it follows by induction that in all rounds until a
broadcast with a Byzantine sender is elected honest parties broadcast
values proposed by honest parties.

Denote by $x_1,x_2,x_3,\ldots$ the number of completed broadcasts in
rounds $r_1,r_2,r_3,\ldots$, respectively.
Thus, for every round $r_i$, $\frac{x_i}{n}$ is the probability to
elect a completed broadcast and fix its decision value, whereas 
$\frac{n - x_i}{n}$ is the probability to elect an uncompleted broadcast
and continue to the next round without fixing a decision.
Let $Q$ be the probability to decide on a value proposed only by
Byzantine parties.
We get that:
\[Q=
\frac{x_1}{n}\frac{f}{x_1} +\frac{n-x_1}{n}(
\frac{x_2}{n}\frac{f}{x_2} + \frac{n-x_2}{n}(
\frac{x_3}{n}\frac{f}{x_3} + \frac{n-x_3}{n}(
\frac{x_4}{n}\frac{f}{x_4} + \frac{n-x_4}{n}(\ldots))))
\]
Let $x = min(\{x_i \mid i \in \mathbb{N}\})$, we get that
\[Q \leq
\frac{x}{n}\frac{f}{x} +\frac{n-x}{n}(
\frac{x}{n}\frac{f}{x} + \frac{n-x}{n}(
\frac{x}{n}\frac{f}{x} + \frac{n-x}{n}(
\frac{x}{n}\frac{f}{x} + \frac{n-x}{n}(\ldots)))) =
\]
\[
\frac{f}{n}(1+\frac{n-x}{n} + (\frac{n-x}{n})^2\ldots) =
\frac{f}{n}(\prod_{i=0}^{\infty} (\frac{n-x}{n})^i ) = 
\frac{f}{n}\frac{n}{n-(n-x)}=\frac{f}{x}
\] 

By the protocol, $x_i \geq n-t$ for every round $r_i$.
Therefore, $Q \leq \frac{f}{n-t}$.

\end{proof}

\paragraph*{Complexity.}
Since the sequential composition uses one instance of RM and one of
AMS, its asymptotic complexity and resilience are equal to that of RM
and AMS, which are proven to be optimal~\cite{FLP85,VABA,bracha1987asynchronous}.
That is, it (1) tolerates up to $t< n/3$ Byzantine parties, (2) runs
in expected constant number of rounds, and (3) sends $O(n^2)$ words
in $O(n^2)$ messages where a word contains a constant number of
signatures and values.

\section{Qualitative Validity vs Weak Validity}

A hasty reader might believe ``Weak Validity (plus agreement) is actually stronger than the first statement of  the Qualitative Validity''  which can lead him to the false conclusion that ``Qualitative Validity is incomparable to Weak Validity.''

Although it might not be obvious in a first glance, Qualitative Validity is strictly stronger than Weak Validity. Namely, (1) every algorithm that satisfies Qualitative Validity (QV) necessarily satisfies Weak Validity, while (2) there are algorithms that satisfy weak validity and do not satisfy QV. Part (2) follows immediately from the second part of the QV definition. Part (1) follows from \Cref{thm:lower}. In particular, \Cref{thm:lower} shows that if the first part of QV is not satisfied and $f > 0$, then there is always a positive probability to decide on a value not in $V_{in}$. For better transparency, we provide here a direct proof of the following:

\begin{claim}
	A protocol that satisfies Qualitative Validity, also satisfies Weak validity.
\end{claim}
\begin{proof}
	Let \alg{} be a protocol that satisfies QV, and let $\ens$ be any ensemble of \alg{} in which $2t+1$ honest parties start with~$v$. We show that~$v$ is the only possible decision in $\ens$. Consider another ensemble $\ens’$ which is the same as $\ens$ except that all corrupted parties (the same as parties in $\ens$) have $v$ as their input but act exactly as in ensemble $\ens$. In ensemble $\ens’$ we have $\maxmult{\Vin} \ge 2t+1+\newf$ and therefore $\maxmult{Vin} -\newf \ge 2t+1$. By the first condition of QV, we get that $v$ is the only possible decision in $\ens’$. Clearly, $\ens$ and $\ens’$ are probabilistically indistinguishable for all honest parties, hence, the only possible decision in $\ens$ is~$v$. 
\end{proof}

\end{document}